\documentclass[letterpaper,12pt]{article}

\usepackage{graphicx}
\graphicspath{ {./images/} }

\usepackage[margin=1in]{geometry}
\usepackage{amsmath,amsthm,amsfonts,amssymb,bm,bbm}

\usepackage[linesnumbered,ruled,vlined]{algorithm2e}

\usepackage{xcolor}
\usepackage{mathtools}
\usepackage{hyperref}
\usepackage{graphicx}
\usepackage{float}  %used for float option [H]
\usepackage{tikz}
\usepackage{authblk}
\usepackage[square,numbers]{natbib}
\newtheorem{theorem}{Theorem}
\newtheorem{lemma}[theorem]{Lemma}
\newtheorem{corollary}[theorem]{Corollary}
\theoremstyle{definition}

\newtheorem{remark}{Remark}
\newtheorem*{example}{Example}

\title{Equivalence of Models of Cake-Cutting Protocols\footnote{This paper is based on the second author's final year project.}}

\author[1]{Paul W. Goldberg\thanks{Paul.Goldberg@cs.ox.ac.uk} \and Ioana Iaru\thanks{iaruioana@gmail.com} \\Department of Computer Science, Oxford University} 

\begin{document}

\maketitle

\begin{abstract}
The cake-cutting problem involves dividing a heterogeneous, divisible resource fairly between $n$ agents.
\citet{BranzeiCKP16}~introduced {\em generalised cut and choose} (GCC) protocols, a formal model for representing
cake-cutting protocols as trees with ``cut'' and ``choose'' nodes corresponding to the agents' actions,
and if-else statements.

In this paper, we identify an alternative and simpler extensive-form game model for cake-cutting protocols,
that we call {\em branch choice} (BC) protocols.
We show that the class of protocols we can represent using this model is invariant
under certain modifications to its definition.
We further prove that any such protocol can be converted to a restricted form in which the agents first cut the cake
and then get to choose between various branches leading to different allocations.
Finally, we show that this model has the same expressive power as GCC protocols,
i.e. they represent the same class of protocols up to a notion of equivalence involving the bounds on
envy that each agent can guarantee for themselves.
For this purpose, we introduce a new notion of envy-equivalence of protocols.
\end{abstract}

\section{Introduction}
The cake-cutting problem is a fair division problem which involves dividing a heterogeneous, divisible resource (the metaphorical ``cake'') between several agents, who have different preferences over different parts of the cake. In practice, this framework can be used to divide various types of resources, such as land or time. The modern study of this problem dates back to \citet{steinhaus} and a description of its history can be found in the books by \citet{brams-taylor} or \citet{robertson-webb}.

To be able to tackle this problem, we need to establish what it means for a division to be fair. There are two common criteria for fairness, \emph{proportionality}, in which each of the $n$ agents receives a piece that they consider to be worth at least $\frac{1}{n}$ of the entire cake, and \emph{envy-freeness}, which means no agent prefers the piece another agent received over their own. If the entire cake is allocated, envy-freeness implies proportionality.

In the literature, well-known cake-cutting protocols are often described in plain language, but to be able to reason about protocols more generally or implement them on a computer, we require some formal model of computation. There are several such models, the most well-known ones being the Robertson-Webb model \citep{robertson-webb}, which involves two types of queries: \emph{Eval} queries, which ask an agent how much they value a given piece of cake, and \emph{Cut} queries, in which an agent specifies where a piece of cake should be cut to be worth a certain value, and the moving-knives model (see, for example, \citet{dubins-spanier}), in which a referee continuously moves one or more knives across the cake until one of the agents calls ``stop''.

More recently, \citet{branzei,BranzeiCKP16} introduced a new model for cake-cutting protocols, \emph{generalised cut and choose} (GCC) protocols, in which rather than reporting their preferences to a referee, the agents divide the cake amongst themselves. In the GCC model, a protocol is represented as a tree with \emph{cut} nodes, in which an agent makes a cut inside one of a set of pre-existing pieces of cake, \emph{choose} nodes, in which an agent chooses from a set of existing pieces and is allocated that piece, and \emph{if-else statements} depending on the execution history of the protocol. 

In this paper, we propose an alternative model for representing cake-cutting protocols as trees, \emph{branch choice} (BC) protocols: at \emph{cut} nodes, an agent makes a cut inside a given piece of cake, similarly to GCC protocols, but at \emph{choose} nodes, an agent instead gets to choose which branch of the tree to proceed to. The resulting pieces are allocated to various agents at leaf nodes. Having such a simple, restrictive model is an advantage when it comes to automation, since making the different possible choices and their outcomes more explicit would make it easier for AI agents to participate in a protocol. 

We define various notions of equivalence between cake-cutting protocols, based on the bounds on value/envy that an agent can guarantee for themselves in the two protocols. We show that making various tweaks to the definition of BC protocols (such as considering protocols that can be represented as directed acyclic graphs rather than trees, or allowing agents to make a cut within a sequence of pieces rather than a single piece), results in the same class of protocols, up to what we call \emph{strong envy-equivalence}. (Informally, an agent can guarantee the same bounds on their envy against the other agents in both protocols.)

We then prove that any BC protocol can be put into a special form in which all the cut nodes come before the choose nodes, making it easier to reason about the structure of the protocol. Finally, we show that in fact BC protocols have the same expressive power as GCC protocols, and we look into how some classic cake-cutting protocols can be represented as BC protocols. 

\subsection{Related Work}
The Robertson-Webb query model, first introduced by \citet{robertson-webb} and formalised by \citet{woeginger}, which is the standard model for discrete cake-cutting protocols in the literature, involves a referee who can ask the agents two types of queries related to their valuation functions (see Introduction), and divides the cake accordingly. A natural question is whether there is an incentive for the agents to report their valuation functions truthfully, i.e. whether there exist any strategy-proof protocols in the Robertson-Webb model. Here, strategy-proofness means truthful reporting is a dominant strategy for every agent, regardless of their valuation function.

\citet{dictatorship} show this is impossible for any ``interesting'' protocols in the Robertson-Webb model: any strategy-proof protocol for two agents is dictatorial (there is a fixed agent who always gets the entire cake), and, more generally, in any strategy-proof protocol for $n$ agents, at least one agent gets an empty piece. There are some other papers that deal with strategy-proofness and the closely related concept of truthfulness. 

Notably, \citet{chen} design a cake-cutting protocol that is truthful, proportional, and envy-free for piecewise-uniform valuation functions, but which uses a direct-revelation mechanism (i.e. the agents reveal their entire valuation function to the referee), rather than the standard query model. A recent paper by \citet{tao} proves that no such protocol can exist for the more general case of strictly positive piecewise-constant valuations, and designs a mechanism that satisfies a weaker notion of truthfulness.

 \citet{branzei,BranzeiCKP16} introduce the GCC protocols model described above, and use it as a framework to design a protocol for which a contiguous allocation is envy-free if and only if it is the outcome of a Nash equilibrium. 
 
The BC protocols model introduced in this paper is different from the GCC model in that instead of having agents choose which pieces of the cake they receive, they choose which branch of the protocol to proceed to. In addition, the agents' actions are more restrictive, so that the ordering of the cuts is encoded in the tree structure, which eliminates the need for if-else statements.

Furthermore, our work does not focus on game-theoretic aspects like strategy-proofness, though this would be an interesting topic to address in the future, but rather on results related to the expressive power of the model, which is a topic that \citet{branzei,BranzeiCKP16} only touch on briefly, and considerations such as the space complexity of implementing common protocols in this format.

For a readable introduction to the better-known cake-cutting protocols, including algorithmic and strategic issues, see \citet{Procaccia13}.

\section{Background}
\subsection{Setting}

The cake, which is represented as the interval $[0,1]$, must be divided between $n$ \emph{agents}.  Each agent must be allocated a \emph{piece of cake} (i.e. a finite union of disjoint intervals). The pieces of cake allocated to different agents must be disjoint. We allow for an individual agent to be given no cake (or an empty piece $[x, x]$), but every piece of cake must be allocated to some agent. 

The agents' preferences are given by \emph{private valuation functions} $V_{1}, V_{2}, ..., V_{n}$ that assign a value $V_i(a, b)$ to every subinterval $[a, b] \subseteq [0, 1]$. We extend the $V_i$ additively to pieces of cake, writing $V_i(X)$ for the value of a piece of cake X. The valuation functions are assumed to have the following properties:
\begin{itemize}
\item \emph{Normalization:} $V_i(0,1) = 1$
\item \emph{Additivity:} For any two disjoint intervals $I_1$ and $I_2$, $V_i(I_1 \cup I_2) = V_i(I_1) + V_i(I_2) $
\item \emph{Divisibility:} For every interval $[a, b] \subseteq [0,1]$ and $0 \leq \lambda \leq 1$ there exists $c \in [a, b]$ such that $V_i(a, c) = \lambda V_i(a, b)$. In particular, this implies the valuation functions are \emph{non-atomic}: $V_i(a, a) = 0$ for every $a \in [0, 1]$.
\item \emph{Non-negativity:} $V_i(a, b) \geq 0$ for any $a, b \in [0,1]$.
\end{itemize}

For an allocation $A = (X_1,...,X_n)$, where $X_i$ is the piece of cake assigned to agent $i$, we define the \emph{envy} of agent $i$ towards agent $j$ as:

 \[envy_A(i, j) = \max(V_i(X_j) - V_i(X_i), 0) \]

\subsection{Generalised Cut and Choose (GCC) Protocols}

A \emph{generalised cut and choose (GCC) protocol} \citep{branzei,BranzeiCKP16}, is represented as a tree in which each node corresponds to the action of an agent. There are three types of nodes:
\begin{itemize}
\item \emph{cut} nodes: an agent makes a cut between two existing cuts.  More specifically, agent $i$ is given a set $S = \{[x_1, y_1], ... , [x_m, y_m]\}$ of intervals representing contiguous piece of cake, such that the endpoints of every piece are either 0,1, or cuts made at a previous step. The agent picks an interval $[x_j, y_j]$ and makes a cut at some point $z$ in it.
\item \emph{choose} nodes: an agent chooses between a set of pieces induced by the existing cuts. Agent $i$ is given a set $S$ as above. They choose an interval $[x_j, y_j]$ from $S$, and this interval then gets allocated to them. In the special case $|S| = 1$, the agent is essentially assigned a piece and there is no actual ``choice''.
\item \emph{if-else} nodes: These nodes can have multiple branches. The protocol progresses to one of the branches based on an if-else statement. The conditions in the if-else statement depend on the order of the cut points made in the previous steps, and the execution history of the protocol (Equivalently, the conditions depend on which pieces the agents cut or chose at each of the previous steps).
\end{itemize}

Note that cut and choose nodes have at most one child each, so the protocol only branches out at if-else nodes. 

In the original paper by \citet{branzei}, the if-else statements are not explicitly included in the tree structure. Instead, the cut/choose nodes have multiple children and progression to one of the children is based on an if-else statement. We view if-else statements as separate nodes for simplicity. It is easy to see that our interpretation and the original one have the same expressive power.

Algorithm \ref{selfridge-conway-gcc} shows how the Selfridge-Conway protocol \citep{robertson-webb}, a classic envy-free protocol for three agents, can be represented as a GCC protocol. Algorithm \ref{selfridge-conway} shows the original protocol, for comparison.

\RestyleAlgo{boxruled}
\LinesNumbered
\begin{algorithm}[ht]
  \caption{The Selfridge-Conway protocol, represented as a GCC protocol. Strictly speaking, wherever the algorithm labels a piece or assumes w.l.o.g. that a certain piece was chosen, we should have an if-else node with multiple branches depending on which piece that is, but we have mostly omitted those for simplicity. \label{selfridge-conway-gcc}}
  CUT: Agent $1$ cuts in $[0,1]$ at point $x_0$.
  
  CUT: Agent $1$ cuts in $[x_0, 1]$ at point $x_1$.
    
  CUT: Agent $2$ cuts a piece from $\{ [0,x_0], [x_0, x_1], [x_1, 1] \}$
  Label this piece $A$, dividing it into pieces $A_1$ and $A_2$. Let the other two pieces be $B$ and $C$. \label{label-piece-a}
  CHOOSE: Agent $3$ chooses a piece from $S = \{A_1, B, C\}$.
  
  \eIf{Agent $3$ chose $A_1$}{
  CHOOSE: Agent $2$ chooses a piece from $S = \{B,C\}$.
  
  CHOOSE: Agent $1$ chooses the last remaining piece in $S$.
  
 CUT: Agent 2 cuts $A_2$ at $\frac{1}{3}$ of the value, resulting in two pieces $A_{21}$, $A_2'$.
	
 CUT: Agent 2 cuts $A_2'$ at $\frac{1}{2}$, resulting in two pieces $A_{22}$, $A_{23}$.
	
 CHOOSE: Agent 3 chooses a piece out of $S' = \{A_{21}, A_{22}, A_{23}\}$, WLOG $A_{21}$.
	
 CHOOSE: Agent $1$ chooses a piece out of $\{A_{22}, A_{23}\}$.
 
 CHOOSE: Agent $2$ chooses the remaining piece.
  }{\tcp {Agent $3$ did not choose $A_1$}
  CHOOSE: Agent $2$ gets piece $A_1$.
  
  CHOOSE: Agent $1$ chooses the last remaining piece in S.
  
 CUT: Agent 3 cuts $A_2$ at $\frac{1}{3}$ of the value, resulting in two pieces $A_{21}$, $A_2'$.
	
 CUT: Agent 3 cuts $A_2'$ at $\frac{1}{2}$, resulting in two pieces $A_{22}$, $A_{23}$.
	
 CHOOSE: Agent 2 chooses a piece out of $S' = \{A_{21}, A_{22}, A_{23}\}$, WLOG $A_{21}$.
	
 CHOOSE: Agent $1$ chooses a piece out of $\{A_{22}, A_{23}\}$.
 
 CHOOSE: Agent 3 chooses the remaining piece.
 }
\end{algorithm}

\RestyleAlgo{boxruled}
\LinesNumbered
\begin{algorithm}[H]
  \caption{The Selfridge-Conway protocol. \label{selfridge-conway}}
  Agent $1$ divides the cake into three pieces of value $\frac{1}{3}.$ 
  
  Let $A$ be the largest piece according to agent $2$. Let the other two pieces be $B$ and $C$.
  
  Agent $2$ trims $A$ so that it is the same size as the second largest piece. Label the trimmed piece $A_1$, and the trimmings $A_2$.
  
  Agent $3$ chooses a piece from $\{A_1, B, C\}$.
    
  \eIf{Agent $3$ chose $A_1$}{
	Agent $2$ chooses one of $B$ and $C$.
  }{\tcp {Agent $3$ did not choose $A_1$}
Agent $2$ gets piece $A_1$.
 }
 Agent $1$ gets the last piece. It remains to divide the trimmings $A_2$.

One of agents $2$ and $3$ got $A_1$. Call that agent $X$, and the other agent $Y$.

 Agent $Y$ cuts $A_2$ into three equal pieces.
 
 Agent $X$ chooses one of the three pieces.
 
 Agent $1$ chooses one of the remaining two pieces.
 
 Agent $Y$ gets the last remaining piece.
\end{algorithm}

\subsection{Equivalence of Cake-Cutting Protocols}
	Since we are working with multiple models for representing cake-cutting protocols, we need to define what it means for two protocols (possibly represented in different models) to be ``the same''. We will present multiple possible notions of equivalence between cake-cutting protocols, which are all based on various bounds agents can guarantee for themselves in terms of value/envy, regardless of the other agents' strategies, so they are independent of the model used to represent a protocol.
	
	 Let $\mathcal{P}$  and $\mathcal{P'}$ be two cake-cutting protocols.
	\begin{enumerate}
	\item \textbf{Value equivalence:} $\mathcal{P}$ and  $\mathcal{P'}$ are equivalent if the minimum value that each agent can guarantee for themselves is the same. That is, if in the worst case for protocol $\mathcal{P}$ agent $i$ gets a piece of value $\lambda$, then in the worst case for protocol $\mathcal{P'}$ the piece agent $i$ gets must have the same value $\lambda$.
	
	\item \textbf{Total envy equivalence:} $\mathcal{P}$ and  $\mathcal{P'}$ are equivalent if the maximum total amount of envy $\sum_{j=1, j \neq i}^{n}envy(i, j)$ an agent can guarantee for themselves is the same in both protocols.
	
	\item \textbf{Pair-wise envy equivalence:} $\mathcal{P}$ and $\mathcal{P'}$ are equivalent if, for any two agents $i$ and $j$, if in one of the protocols agent $i$ can guarantee $envy(i, j) \leq M$, for some $M \in [0,1]$, then they can guarantee the same bound holds in the other protocol.
		
	\item \textbf{Strong envy equivalence:} $\mathcal{P}$ and $\mathcal{P'}$ are equivalent if, for any agent $i$ and any set $S$ of agents not containing $i$, if in one of the protocols agent $i$ can guarantee simultaneous bounds $envy(i, j) \leq M_j, \forall j \in S$, for some $M_j \in [0, 1]$, then they can guarantee the same bounds hold in the other protocol. This is the notion of equivalence we will be using the most throughout the paper. Note this is still quite broad, since, in particular, any two envy-free protocols will be strongly envy-equivalent. A lot of our results involve protocols that appear ``similar'' in a stronger sense (for example, some of our conversion algorithms involve mimicking the actions in a given protocol in a different model), but this is difficult to formalise and since, in practice, most protocols aim to guarantee certain bounds on envy, our definition is good enough.
	
	Note also that strong envy equivalence implies pair-wise envy equivalence (taking $S = \{ j \}$).
	\end{enumerate}

\section{An Alternate Model for Cake-Cutting Protocols}

In Section~\ref{sec:bcdef} we define the class of Branch Choice (BC) protocols.
In Sections~\ref{sec:bcdag} and~\ref{sec:bcext} we consider variations of the definition, and we prove that they are equivalent to the original one. In Section~\ref{sec:bccuts} we show that no expressive power is lost in moving to a restricted version in which the ``cut'' nodes must precede the ``choose'' nodes. In Section~\ref{sec:bcgcc} we show that BC protocols have the same expressive power as GC protocols. In Section~\ref{sec:bcprop} we show that certain proportional protocols can be expressed as BC (equivalently GCC) protocols.

\subsection{Branch Choice (BC) Protocols}\label{sec:bcdef}

Consider the following alternative representation of cake-cutting protocols. As before, we represent the protocols as trees, but this time we divide the nodes into:
\begin{itemize}
\item non-leaf nodes, which are further divided into:
	\begin{itemize}
	\item \emph{cut} nodes: We have a partition of the cake into $m$ contiguous pieces, ordered from left to right.  Agent $i$ subdivides a piece $j$ by making one cut. Trivial cuts (i.e. cutting at one of the ends of the piece) are allowed. For convenience, we consider that for a piece $[x, y]$, cutting at one of the ends, WLOG $x$, divides it into $[x, y]$ and a new piece $[x, x]$. Such ``empty'' pieces will have value $0$ to any agent, but they allow us to keep track of the number of pieces throughout the tree. Note $i$ and $j$ are both fixed, unlike in the previous section. 
	\item \emph{choose} nodes: Given a partition of the cake as above, agent $i$ is allowed to choose which child of the current node to proceed to, based on the position of the cuts.
	\end{itemize}
\item leaf nodes: Given a partition of the cake, each piece in the partition is allocated to an agent. The allocation is represented as a mapping $\{1,..., m\} \rightarrow \{1, ..., n\}$ where $n$ is the number of agents, and $m$ is the number of pieces in the partition. Note that for a given leaf, the number $m$ of pieces the cake has been divided into by the time we reach that leaf is one plus the number of cut nodes above the leaf, which does not depend on the execution history, so the mapping is well-defined.
\end{itemize}

We say a cut/choose node corresponding to an action of agent $i$ is \emph{controlled by agent} $i$.

We will refer to the class of protocols that can be represented in the above form as \emph{branch choice (BC) protocols} to reflect the fact that at choose nodes, an agent chooses which branch of the tree is executed.

Note we need to keep track of the execution history at every node (i.e. how the cake is partitioned before the action corresponding to the current node). We can either store this information at every node (requiring extra space), or recover it from the unique path between the current node and the root (no extra space). Both take linear time in the number of nodes to read, so we will use the latter option and store no extra information at the nodes.

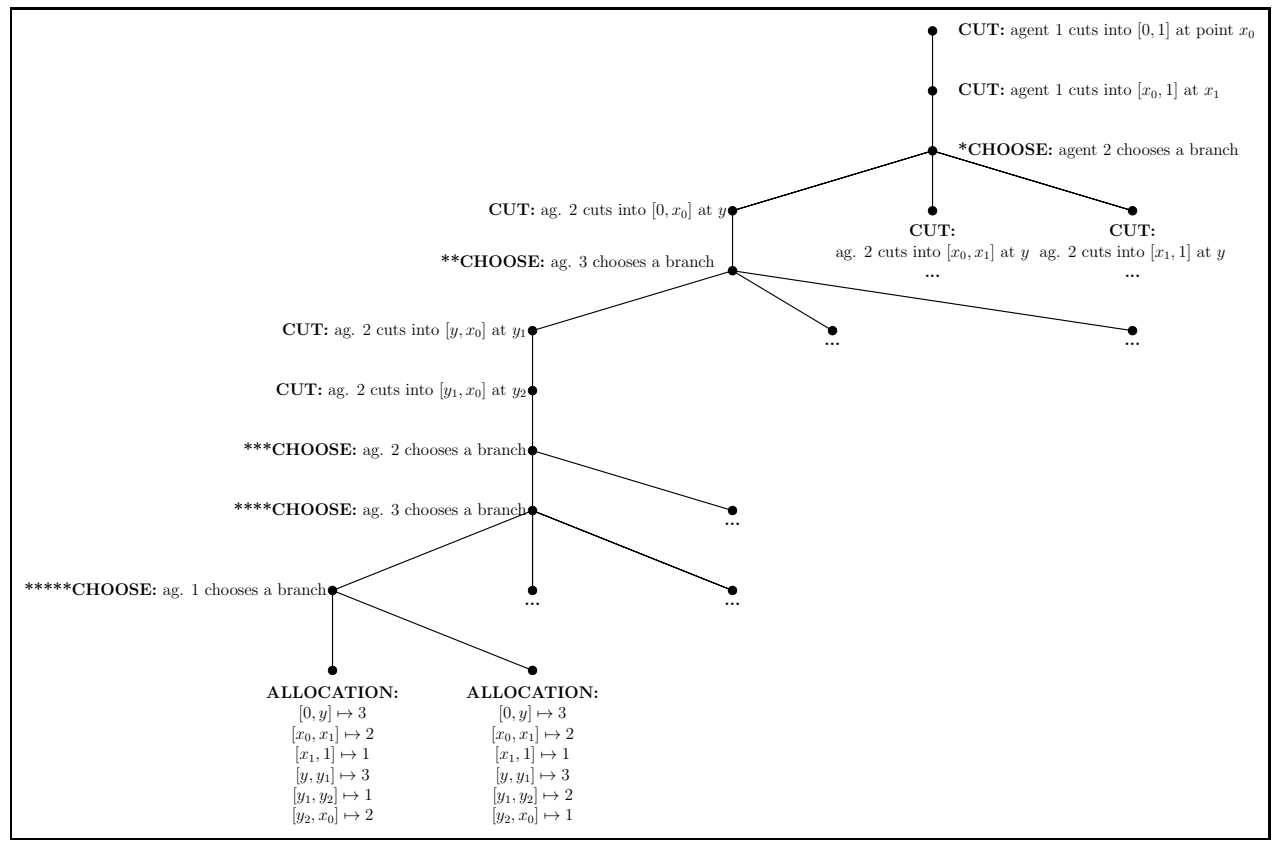
\begin{figure}[H]
\fbox{%
\resizebox{\textwidth}{!}{%
\begin{tikzpicture}

\draw[fill=black] (0,0) circle (3pt);
\draw[fill=black] (0,-1.5) circle (3pt);
\draw[fill=black] (0,-3) circle (3pt);
\draw[fill=black] (-5,-4.5) circle (3pt);
\draw[fill=black] (0,-4.5) circle (3pt);
\draw[fill=black] (5, -4.5) circle(3pt);
\draw[fill=black] (-5,-6) circle (3pt);
\draw[fill=black] (-10,-7.5) circle (3pt);
\draw[fill=black] (-2.5,-7.5) circle (3pt);
\draw[fill=black] (5,-7.5) circle (3pt);
\draw[fill=black] (-10, -9) circle(3pt);
\draw[fill=black] (-10, -10.5) circle(3pt);
\draw[fill=black] (-10, -12) circle(3pt);
\draw[fill=black] (-5, -12) circle(3pt);
\draw[fill=black] (-5, -14) circle(3pt);
\draw[fill=black] (-10, -14) circle(3pt);
\draw[fill=black] (-15, -14) circle(3pt);
\draw[fill=black] (-15, -16) circle(3pt);
\draw[fill=black] (-10, -16) circle(3pt);

\node[anchor=west] at (0.5,0)  {\textbf{CUT:} agent $1$ cuts into $[0, 1]$ at point $x_0$};
\node[anchor=west] at (0.5, -1.5) {\textbf{CUT:} agent $1$ cuts into $[x_0, 1]$ at $x_1$};
\node[anchor=west] at (0.5, -3)  
{\textbf{*CHOOSE:} agent $2$ chooses a branch};

\node[anchor=east] at (-5, -4.5) {\textbf{CUT:} ag. $2$ cuts into $[0, x_0]$ at $y$};

\node[align=center, anchor=south] at (0, -6.3) {\textbf{CUT:} 
\\
ag. $2$ cuts into $[x_0, x_1]$ at $y$
\\
\textbf{...}
};

\node[align=center, anchor=south] at (5, -6.3) {\textbf{CUT:}
\\
ag. $2$ cuts into $[x_1,1]$ at $y$
\\
\textbf{...}
};

\node[anchor=east] at (-5.3, -5.8) {\textbf{**CHOOSE:} ag. $3$ chooses a branch};
\node[anchor=south] at (-2.5, -8) {\textbf{...}};
\node[anchor=south] at (5, -8) {\textbf{...}};

\node[anchor=east] at (-10, -7.5) {\textbf{CUT:} ag. $2$ cuts into $[y, x_0]$ at $y_1$};
\node[anchor=east] at (-10, -9) {\textbf{CUT:} ag. $2$ cuts into $[y_1, x_0]$ at $y_2$};

\node[anchor=east] at (-10, -10.5) {\textbf{***CHOOSE:} ag. $2$ chooses a branch};
\node[anchor=east] at (-10, -12) {\textbf{****CHOOSE:} ag. $3$ chooses a branch};
\node[anchor=south] at (-5, -12.5) {\textbf{...}};

\node[anchor=south] at (-10, -14.5) {\textbf{...}};
\node[anchor=south] at (-5, -14.5) {\textbf{...}};

\node[anchor=east] at (-15, -14) {\textbf{*****CHOOSE:} ag. $1$ chooses a branch};

\node[align=center, anchor=south] at (-15, -20) {\textbf{ALLOCATION:}
\\
$[0,y] \mapsto 3$
\\
$[x_0, x_1] \mapsto 2$
\\
$[x_1, 1] \mapsto 1$
\\
$[y, y_1] \mapsto 3$
\\
$[y_1, y_2] \mapsto 1$
\\
$[y_2, x_0] \mapsto 2$
};

\node[align=center, anchor=south] at (-10, -20) {\textbf{ALLOCATION:}
\\
$[0,y] \mapsto 3$
\\
$[x_0, x_1] \mapsto 2$
\\
$[x_1, 1] \mapsto 1$
\\
$[y, y_1] \mapsto 3$
\\
$[y_1, y_2] \mapsto 2$
\\
$[y_2, x_0] \mapsto 1$
};

\draw[thick] (0,0) -- (0,-1.5) -- (0,-3) -- (-5, -4.5) -- (0, -3) -- (5, -4.5) -- (0, -3) -- (0, -4.5);
\draw[thick] (-5,-4.5) -- (-5, -6);
\draw[thick] (-5,-6) -- (-10, -7.5);
\draw[thick] (-5,-6) -- (-2.5, -7.5);
\draw[thick] (-5,-6) -- (5, -7.5);
\draw[thick] (-10, -12) -- (-10, -10.5) -- (-10,-9) -- (-10, -7.5);
\draw[thick] (-5, -12) -- (-10, -10.5);
\draw[thick] (-10, -12) -- (-10, -14) -- (-10,-12) -- (-5,-14) -- (-10, -12) -- (-15, -14);
\draw[thick] (-10, -16) -- (-15, -14) -- (-15, -16);

\end{tikzpicture}}}
\caption{(Partial) tree representation of a BC protocol that is strongly envy-equivalent to the Selfridge-Conway protocol (i.e. it is envy-free). The omitted subtrees, represented by ``...'', are essentially copies of the one shown, for some permutation of the pieces/agents. The complete tree would have $150$ nodes in total. See detailed explanation below.}
 \label{fig:selfridge-conway-bc}
\end{figure}

In figure \ref{fig:selfridge-conway-bc}:

\textbf{*} - agent $2$ trims the largest piece according to them (as per the original protocol). In our representation of the protocol, they have to choose the largest piece to minimise their envy.

\textbf{**} - agent $3$ chooses which of the three ``big'' pieces $[0,y], [x_0, x_1]$ and $[x_1, 1]$ they get.

\textbf{***} - agent $2$ chooses which of the remaining two ``big'' pieces they get. In the two branches we have left out, where agent $1$ hasn't chosen $[0,y]$, agent $2$ is assigned $[0,y]$ by default and does not get a choice here.

\textbf{****}, \textbf{*****} - agents $3$ and $1$ choose which of the ``trimmings'' they are assigned.

It can be shown that all other conditions in the original protocol (e.g. agent $1$ cutting the piece into thirds at the beginning, agent $2$ trimming the largest piece to be the same size as the second largest, etc.) can be recovered by analysing where the agents need to cut to minimise their envy.

\subsection{BC Protocols Represented as Directed Acyclic Graphs}\label{sec:bcdag}
As a generalisation, we can instead consider the class of cake-cutting protocols that can be represented as \emph{directed acyclic graphs} (DAGs) using the same types of nodes as BC protocols. 

A DAG representing a cake-cutting protocol will have a designated \emph{root} node that corresponds to the first action in the protocol. We will also assume that every node is reachable from the root, as nodes that are not reachable will not be relevant to the protocol. To ensure the DAG represents a valid protocol, we must also require that if there are multiple paths from the root to a node, all of them result in the same number of cuts, so that the number of pre-existing cuts at a given node is well-defined and we can describe actions like ``agent $i$ cuts into the $j$-th piece''. This can provide a more compact representation for certain protocols, but we will see that it is no more ``powerful'' than the tree representation, that is, given a DAG representing some protocol, we can represent the same protocol as a tree.

\begin{theorem}\label{DAGs to trees}
The protocols that can be represented as DAGs are exactly the BC protocols as defined above.
\end{theorem}
\begin{proof}
Clearly any BC protocol can be represented as a DAG by converting every edge in the tree representation to a directed parent $\rightarrow$ child edge.

Conversely, let $G$ be a DAG representing a cake-cutting protocol. We will produce a tree representation of the same protocol. The root and nodes with a single parent can be translated directly into the tree. For nodes with multiple parents, the key idea is to make one copy of the node and its descendants for each parent. Every path from the root to the node in the initial graph $G$ corresponds to a path from the root to one of the copies of the node in the resulting graph. Therefore the new graph represents the same protocol. By repeating this step until there are no more nodes with multiple parents, we will obtain a tree representation of the protocol. 

This algorithm will always terminate, but to eliminate repeated work we want to ensure that a node's descendants are converted to tree nodes before the node itself (otherwise we might have to run the process multiple times for copies of the same descendant). We can do this by finding a topological ordering of the tree, and then converting the nodes in reverse topological order.

Therefore, any protocol represented as a DAG can also be represented as a tree, so the protocols that can be represented as DAGs are precisely the BC protocols.
\end{proof}

\subsection{Extended BC Protocols}\label{sec:bcext}

We will now consider a modified definition of BC protocols, which allows an agent to choose which piece of cake to cut from a sequence of consecutive pieces, and allows for allocating a sequence of pieces to some agent at once by only specifying the endpoints of the sequence.

The advantage of this more flexible definition is that we can ignore some extraneous cuts in the execution of the protocol, as we will see, for example, in the proof of theorem \ref{cuts-before-choices}. We will show that, in fact, this extended definition has the same expressive power as the original one.

\begin{itemize} 
\item Instead of having some agent $i$ cut into a single piece $[x, y]$ with no other cuts inside it, we allow for \emph{cut} nodes that ask agent $i$ to make a cut between two existing cuts $x$ and $y$ (which might not be consecutive). 
\item At \emph{leaves}, instead of allocating each interval, we allow for allocating the piece of cake between two (potentially non-consecutive) cuts $x$ and $y$ to some agent, with the constraint that $x < y$ (so if the ordering of $x$ and $y$ cannot be determined solely from the structure of the protocol, and depends on the decisions of the agents, this is not a valid allocation). Note this is essentially the same as BC protocol allocation, except the allocation can be written more concisely, ignoring intermediate cuts if consecutive intervals are allocated to the same agent.

\begin{example} 
If agent $1$ makes a cut at $x \in [0,1]$ and agent $2$ then makes a cut at $y \in [0,1]$, we cannot have an allocation $[0,x] \mapsto 1$, $[x, y] \mapsto 2$, $[y,1] \mapsto 1$, because it might be the case that $y < x$. But if instead we restrict the second cut to $y \in [x,1]$, the allocation is valid. In both cases,  $[0,x] \mapsto 1$, $[x, 1] \mapsto 2$ is also a valid allocation.
\end{example}
\end{itemize}

We will refer to protocols covered by this definition as \emph{extended BC protocols}. Note that indeed any protocol covered by the original definition of BC protocols is covered by the modified definition, with no changes needed. 

\begin{theorem}\label{extended-bc-protocols}
For any extended BC protocol, we can construct a strongly envy-equivalent BC protocol. This results in a $\Theta(n!)$ blowup in the size of the tree representation of the protocol.
\end{theorem}

\begin{proof}

\textbf{Construction:} We convert the cut nodes in topological order, so when we reach a node, we know that all of its descendants are already ``valid'' nodes in the original definition. For a cut node ``agent $i$ cuts between $x$ and $y$'', suppose there are $n$ other cuts $x_1, ..., x_n$ between $x$ and $y$. Note $n$ is fixed, because we have already converted all of the current node's descendants, so we know the ordering of the pre-existing cuts at this stage. We replace our cut node with a choose node with $n$ cut node children, corresponding to cutting in $[x, x_1], [x_1, x_2],..., [x_n, y]$, respectively. Each child will have a copy of the initial cut node's subtree. 

We also need to modify the leaf nodes. After modifying the cut nodes, the ordering of the cuts at each leaf node is now fixed, so we can replace an allocation $[x, y] \mapsto j$ with one that allocates each continuous piece between $x$ and $y$ to agent $j$. 

\textbf{Complexity analysis:} This conversion process terminates, but it results in an $\Theta(n!)$ blowup in the size of the tree (where n is the initial number of nodes). Consequently, the running time is also $\Theta(n!)$.

	To see this, consider the worst case, in which we have a chain of $n$ cut nodes that are all of the form ``agent $j$ cuts in $[0,1]$ at $x_i$'', with $1 \leq i \leq n$). 
	
	For the first cut node we don't have to do anything. For the second cut node, we replace it with a choose node with $2$ branches corresponding to cuts in $[0, x_0]$ and $[x_0, 1]$, respectively.
	 
	In each of these two branches, there will be a copy of the third cut node, which needs to be replaced with a choose node with $3$ branches, etc.
	
	See figures \ref{fig:original-protocol} and \ref{fig:resulting-protocol} for an example.

\begin{figure}[H]
\centering
\fbox{%
\begin{tikzpicture}

\draw[fill=black] (0,0) circle (3pt);
\draw[fill=black] (0,-1.5) circle (3pt); 
\draw[fill=black] (0,-3) circle (3pt);

\node[anchor=west] at (0.5,0) {\textbf{CUT:} agent $1$ cuts $[0, 1]$ at $x$};
\node[anchor=west] at (0.5, -1.5) {\textbf{CUT:} agent $2$ cuts $[0,1]$ at $y$};
\node[anchor=west] at (0.5, -3) {\textbf{CUT:} agent $3$ cuts $[0,1]$ at $z$};

\draw[thick] (0,0) -- (0,-1.5) -- (0,-3);
\end{tikzpicture}}
\caption{Worst-case for n=3, before conversion.}
 \label{fig:original-protocol}
\end{figure}
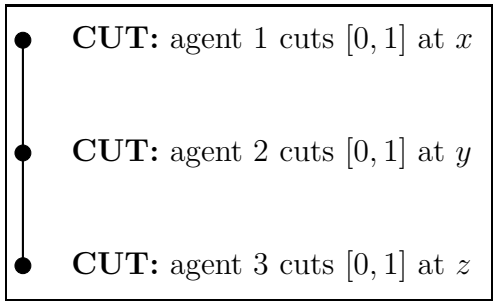

\begin{figure}[H]
\centering
\fbox{%
\resizebox{\textwidth}{!}{%
\begin{tikzpicture}

\draw[fill=black] (0,0) circle (3pt);
\draw[fill=black] (0,-1.5) circle (3pt); 
\draw[fill=black] (-5,-3) circle (3pt);
\draw[fill=black] (5,-3) circle (3pt);
\draw[fill=black] (-5,-5) circle (3pt);
\draw[fill=black] (5,-5) circle (3pt);
\draw[fill=black] (-10,-7) circle (3pt);
\draw[fill=black] (-5,-7) circle (3pt);
\draw[fill=black] (0,-7) circle (3pt);
\draw[fill=black] (5,-7) circle (3pt);
\draw[fill=black] (10,-7) circle (3pt);
\draw[fill=black] (15,-7) circle (3pt);

\node[anchor=west] at (0.5,0) {\textbf{CUT:} agent $1$ cuts into $[0, 1]$ at $x$};
\node[anchor=west] at (0.5, -1.5) {\textbf{CHOOSE:} agent $2$ chooses a branch};
\node[anchor=west] at (-5, -3.2) {\textbf{CUT:} agent $2$ cuts into $[0,x]$ at $y$};
\node[anchor=west] at (5, -3) {\textbf{CUT:} agent $2$ cuts into $[x,1]$ at $y$};

\node[anchor=west] at (-5, -4.8) {\textbf{CHOOSE:} agent $3$ chooses a branch};
\node[anchor=west] at (5, -4.8) {\textbf{CHOOSE:} agent $3$ chooses a branch};

\node[align=center,anchor=south] at (-10, -8.5) {\textbf{CUT:} 
\\
agent $3$ cuts into $[0,y]$ at $z$};
\node[align=center, anchor=south] at (-5, -8.5) {\textbf{CUT:} 
\\
agent $3$ cuts into $[y, x]$ at $z$};
\node[align=center, anchor=south] at (0, -8.5) {\textbf{CUT:}
\\
agent $3$ cuts into $[x,1]$ at $z$};
\node[align=center, anchor=south] at (5, -8.5) {\textbf{CUT:} 
\\
agent $3$ cuts into $[0,x]$ at $z$};
\node[align=center, anchor=south] at (10, -8.5) {\textbf{CUT:} 
\\
agent $3$ cuts into $[x, y]$ at $z$};
\node[align=center, anchor=south] at (15, -8.5) {\textbf{CUT:} 
\\
agent $3$ cuts into $[y, 1]$ at $z$};

\draw[thick] (0,0) -- (0,-1.5) -- (-5,-3) -- (0, -1.5) -- (5,-3);
\draw[thick] (-5,-3) -- (-5, -5);
\draw[thick] (5,-3) -- (5, -5);
\draw[thick] (-5,-5) -- (-10, -7) -- (-5,-5) -- (-5, -7) --  (-5,-5) -- (0, -7);
\draw[thick] (5,-5) -- (10, -7) -- (5,-5) -- (5, -7) --  (5,-5) -- (15, -7);

\end{tikzpicture}}}
\caption{Worst-case for n=3, after conversion.}
 \label{fig:resulting-protocol}
\end{figure}
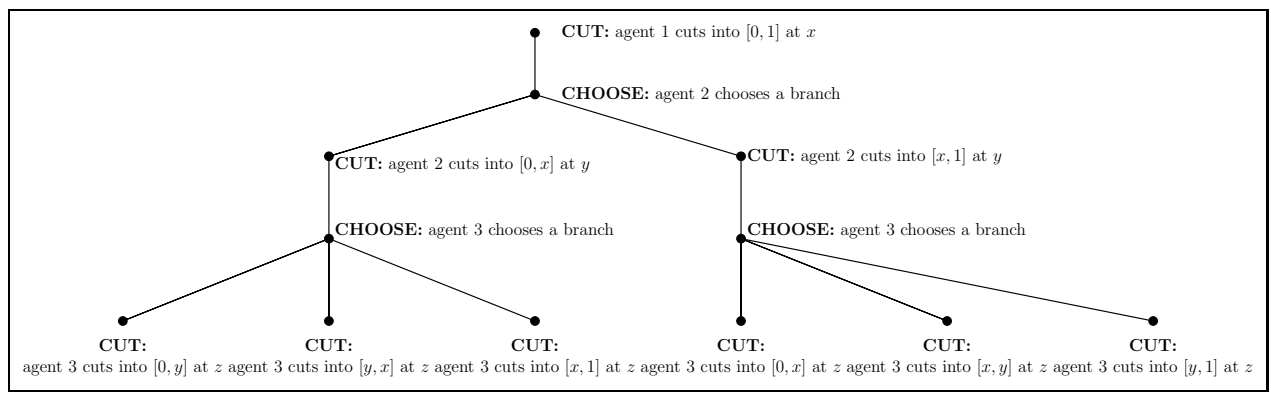

\end{proof}

\begin{corollary} \label{extended-bc-protocols-equivalence}
BC protocols and extended BC protocols represent the same class of protocols, up to strong envy-equivalence. 
\end{corollary}

\subsection{Converting BC Protocols to ``Cuts Before Choices'' Form}\label{sec:bccuts}
In this section, we will prove that any BC protocol can be converted to a ``nice'' form in which all the cut nodes come before all the choose nodes, so, intuitively, the agents first cut up the cake, and then they make choices between branches leading to various allocations. 

\begin{theorem}\label{cuts-before-choices}
Any extended BC protocol can be converted into a strongly envy-equivalent extended BC protocol in which no choose node in the tree representation of the protocol has a cut node descendant. The conversion algorithm runs in polynomial time, and the size (number of nodes) of the original tree is equal to the size of the resulting tree. 
\end{theorem}

\begin{proof}

Let $\mathcal{P}$ be a cake-cutting protocol. If the protocol already satisfies the required condition, we are done. Otherwise, suppose the extended BC tree representation of $\mathcal{P}$ contains a choose node with a cut node descendant. In particular, there must be a choose node with a cut node child ``agent $i'$ cuts between $x_k$ and $x_{k+1}$ at point $y$''.

We will describe an algorithm to move the cut node above the choose node (as in figures \ref{fig:before} and \ref{fig:after}) without affecting the resulting allocation (each agent will be assigned the same piece of cake as in the initial protocol, except potentially cut up differently). The advantage of this form is that it makes it easier to reason about the protocol.

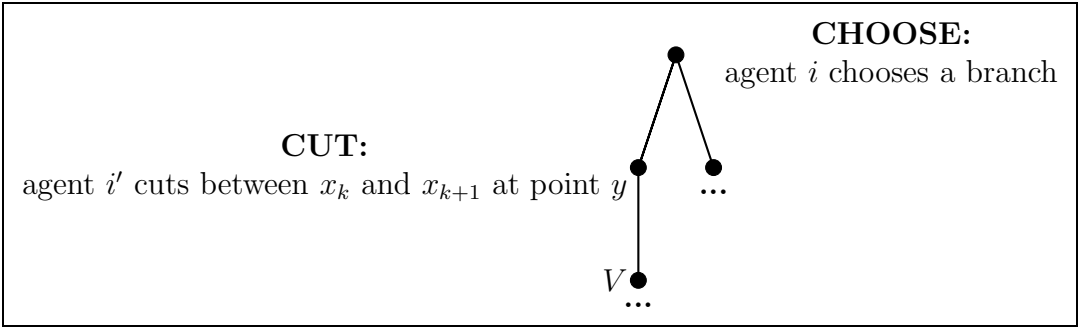
\begin{figure}[H]
\centering
\fbox{%
\begin{tikzpicture}

\draw[fill=black] (0,-3) circle (3pt);
\draw[fill=black] (-0.5,-4.5) circle (3pt);
\draw[fill=black] (-0.5, -6) circle(3pt);
\draw[fill=black] (0.5, -4.5) circle(3pt);

\node[align=center, anchor=west] at (0.5, -3) (choose) {\textbf{CHOOSE:} \\ agent $i$ chooses a branch};

\node[align=center, anchor=east] at (-0.5, -4.5) (cut-branch) {\textbf{CUT:} 
\\
agent $i'$ cuts between $x_k$ and $x_{k+1}$ at point $y$};

\node[align=center, anchor=south] at (0.5, -5) (other-branches) {\textbf{...}};

\node[align=center, anchor=east] at (-0.5, -6) (cut-nodes-child) {$V$};

\node[align=center, anchor=south] at (-0.5, -6.5) (V-branch) {\textbf{...}};

\draw[thick] (0,-3) -- (-0.5, -4.5) -- (0, -3) -- (0.5, -4.5);
\draw[thick] (-0.5, -4.5) -- (-0.5, -6);
\end{tikzpicture}}
\caption{Choose node with cut node child.}
 \label{fig:before}
\end{figure}

\begin{figure}[h]
\centering
\fbox{%
\begin{tikzpicture}

\draw[fill=black] (0, -1.5) circle (3pt);
\draw[fill=black] (0,-3) circle (3pt);
\draw[fill=black] (-0.5,-4.5) circle (3pt);
\draw[fill=black] (0.5, -4.5) circle(3pt);

\node[align=center, anchor=west] at (0.5, -3) (choose) {\textbf{CHOOSE:} agent $i$ chooses a branch};

\node[align=center, anchor=west] at (0.5, -1.5) (cut-branch) {\textbf{CUT:} agent $i'$ cuts between $x_k$ and $x_{k+1}$ at point $y$};

\node[align=center, anchor=south] at (0.5, -5) (other-branches) {\textbf{...}};

\node[align=center, anchor=east] at (-0.5, -4.5) (cut-nodes-child) {$V$};

\node[align=center, anchor=south] at (-0.5, -5) (V-branch) {\textbf{...}};

\draw[thick] (0, -1.5) -- (0,-3) -- (-0.5, -4.5) -- (0, -3) -- (0.5, -4.5);
\end{tikzpicture}}
\caption{Subtree obtained by moving the cut node above the choose node in \ref{fig:before}.}
 \label{fig:after}
\end{figure}
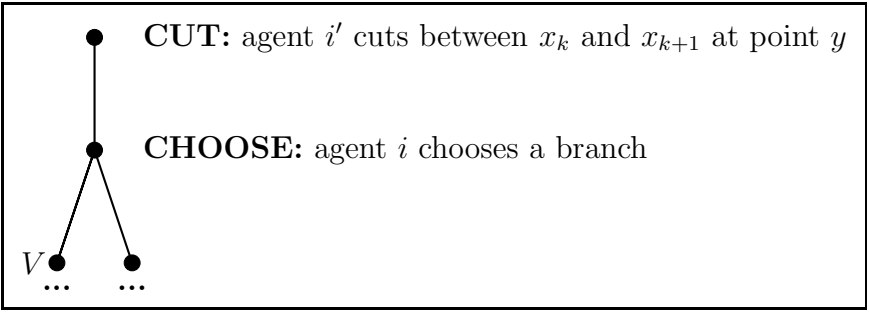

Note the cut node only has one child, call it $V$, which can be any kind of node. Moving the cut node above the choose node will not affect the subprotocol defined by $V$'s branch. In the other branches, we just ignore the cut at $y$ (which extended BC protocols allow) and execute the protocol exactly as before. 

We repeat this process until there is no choose node with a cut node child. This protocol will contain no choose node with a cut node descendant. Since the conversion algorithm does not affect the allocation, any simultaneous bounds on an agent's envy towards other agents in the initial protocol will carry over to the resulting protocol, so the two are strongly envy-equivalent. 

Furthermore, the running time of the conversion algorithm is polynomial (we move each cut node up the tree at most $n$ times, where $n$ is the total number of nodes, and there are at most $n$ cut nodes).

Since no nodes are added and deleted, the size of the resulting tree is equal to the size of the original tree. 
\end{proof}

\begin{lemma}\label{bc-protocols-intermediate-form}
Any BC protocol can be converted into a strongly envy-equivalent BC protocol in which every choose node either has:
\begin{itemize}
\item no cut node descendants, or
\item only cut node children, in which case if the choose node is controlled by agent $i$, the cut node children are guaranteed to be controlled by agent $i$ and correspond to cuts in a sequence of consecutive pieces.
\end{itemize}
The conversion results in a $\Theta(n!)$ blowup in the size of the tree representation.
\end{lemma}

\begin{proof}
Let $\mathcal{P}$ be a BC protocol. In particular,  $\mathcal{P}$ can be viewed as an extended BC protocol, so by theorem \ref{cuts-before-choices} we can convert it to a strongly envy-equivalent ``cuts before choices'' extended BC protocol $\mathcal{P}'$ in which no choose node in the tree representation of the protocol has a cut node descendant, resulting in a tree similar to the one in the figure below.

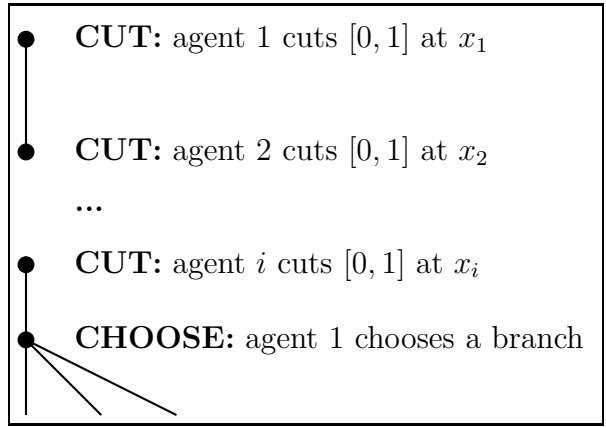
\begin{figure}[H]
\centering
\fbox{%
\begin{tikzpicture}

\draw[fill=black] (0,0) circle (3pt);
\draw[fill=black] (0,-1.5) circle (3pt); 
\draw[fill=black] (0,-3) circle (3pt);
\draw[fill=black] (0, -4) circle(3pt);

\node[anchor=west] at (0.5,0) {\textbf{CUT:} agent $1$ cuts $[0,1]$ at $x_1$};
\node[anchor=west] at (0.5, -1.5) {\textbf{CUT:} agent $2$ cuts $[0,1]$ at $x_2$};
\node[anchor=west] at (0.5, -2.25) {\textbf{...}};
\node[anchor=west] at (0.5, -3) {\textbf{CUT:} agent $i$ cuts $[0,1]$ at $x_i$};
\node[anchor=west] at (0.5, -4) {\textbf{CHOOSE:} agent $1$ chooses a branch};

\draw[thick] (0,0) -- (0,-1.5);
\draw[thick] (0,-3) -- (0,-4);
\draw[thick] (0,-4) -- (0, -5); 
\draw[thick] (0,-4) -- (1, -5); 
\draw[thick] (0,-4) -- (2, -5); 

\end{tikzpicture}}
 \label{fig:intermediate-step1}
\end{figure}

We can then convert $\mathcal{P}'$ back to a BC protocol using the algorithm in theorem \ref{extended-bc-protocols}.

In the last conversion, a cut node ``agent $i$ cuts between $x_0$ and $x_n$'' in $\mathcal{P}'$ will either be unchanged, if there are no other cuts between $x_0$ and $x_n$, or it will be split into a choose node ``agent $i$ chooses a branch'' with $n$ cut node children ``agent $i$ cuts between $x_j$ and $x_{j+1}$'' corresponding to each of the intermediate pieces (where $x_1,...,x_{n-1}$ are all the existing cuts between $x_0$ and $x_n$, from left to right).

The choose nodes in $\mathcal{P'}$ will not be changed/moved, so they will still have no cut node descendants after the conversion.

Therefore the resulting protocol has the required form, and is strongly envy-equivalent to the original protocol $\mathcal{P}$. The $\Theta(n!)$ blowup in size comes from the last conversion. 
\end{proof}

\begin{theorem}\label{bc-cuts-before-choices}
Any BC protocol can be converted into a strongly envy-equivalent BC protocol in which no choose node in the tree representation of the protocol has a cut node descendant. 
\end{theorem}

\begin{proof}
Let $\mathcal{P}$ be a BC protocol.

We start by converting $\mathcal{P}$ to the form in lemma \ref{bc-protocols-intermediate-form}. The resulting tree will look like in figure \ref{fig:intermediate-step2}: a chain of $m-1$ cut nodes, followed by a choose node with $\leq m$ branches, then a level with $\leq m$ cut nodes, one in each branch, then $\leq m$ choose nodes each with $\leq m+1$ branches, etc., until the last level of cut nodes (whose choose node parents will each have $\leq m+k$ branches, where $k$ is the number of levels of cut nodes). Below that level, the tree only contains choose nodes and allocation nodes.

\begin{figure}[H]
\centering
\fbox{%
\resizebox{\textwidth}{!}{%
\begin{tikzpicture}

\draw[fill=black] (0,3) circle (3pt);
\draw[fill=black] (0,1.5) circle (3pt);
\draw[fill=black] (0,0) circle (3pt);
\draw[fill=black] (0,-1.5) circle (3pt); 
\draw[fill=black] (-5,-3) circle (3pt);
\draw[fill=black] (5,-3) circle (3pt);
\draw[fill=black] (-5,-5) circle (3pt);
\draw[fill=black] (5,-5) circle (3pt);
\draw[fill=black] (-11,-7) circle (3pt);
\draw[fill=black] (-5,-7) circle (3pt);
\draw[fill=black] (0,-7) circle (3pt);
\draw[fill=black] (5,-7) circle (3pt);
\draw[fill=black] (11,-7) circle (3pt);
\draw[fill=black] (17,-7) circle (3pt);

\node[anchor=west] at (0.5,3) {\textbf{CUT} at $x_1 \in [0, 1]$};
\node[anchor=west] at (0.5,1.5) {\textbf{CUT} at $x_2 \in [x_1, 1]$};
\node[anchor=west] at (0.5,0.75) {\textbf{...}};
\node[anchor=west] at (0.5,0) {\textbf{CUT} at $x_{m-1} \in [x_{m-2}, 1]$};
\node[anchor=west] at (3, 1) {\textbf{$m-1$ cut nodes}};

\node[anchor=west] at (0.5, -1.5) {\textbf{CHOOSE:} agent i chooses a branch \textbf{($\leq m$ branches)}};
\node[align=center, anchor=south] at (0, -2.5) {\textbf{.    .    .}};
\node[anchor=west] at (-5, -3.2) {\textbf{CUT:} agent $i$ cuts at $y_0 \in [0,x_1]$};
\node[anchor=west] at (5, -3) {\textbf{CUT:} agent $i$ cuts at $y_m \in [x_{m-1}, 1]$};

\node[anchor=west] at (-5, -4.8) {\textbf{CHOOSE:} agent $j$ chooses \textbf{($\leq m+1$ branches)}};
\node[anchor=west] at (5, -4.8) {\textbf{CHOOSE:} agent $j$ chooses \textbf{($\leq m+1$ branches)}};

\node[align=center,anchor=south] at (-10, -8.5) {\textbf{CUT:} 
\\
agent $j$ cuts at $z_0 \in [0, y_0]$};
\node[align=center, anchor=south] at (-5, -8.5) {\textbf{CUT:} 
\\
agent $j$ cuts at $z_1 \in [y_0, x_1]$};
\node[align=center, anchor=south] at (0, -8.5) {\textbf{CUT:}
\\
agent $j$ cuts at $z_m \in [x_{m-1}, 1]$};

\node[align=center, anchor=south] at (-4, -6.5) {\textbf{.    .    .}};

\node[align=center, anchor=south] at (5, -8.5) {\textbf{CUT:} 
\\
agent $j$ cuts at $z_0' \in [0, x_1]$};

\node[align=center, anchor=south] at (6, -6.5) {\textbf{.    .    .}};

\node[align=center, anchor=south] at (11, -8.5) {\textbf{CUT:} 
\\
agent $j$ cuts at $z_{m-1}' \in [x_{m-1}, y_m]$};

\node[align=center, anchor=south] at (17, -8.5) {\textbf{CUT:} 
\\
agent $j$ cuts at $z_m' \in [y_m, 1]$};

\draw[thick] (0,3) -- (0,1.5) -- (0,0) -- (0,-1.5) -- (-5,-3) -- (0, -1.5) -- (5,-3);
\draw[thick] (-5,-3) -- (-5, -5);
\draw[thick] (5,-3) -- (5, -5);
\draw[thick] (-5,-5) -- (-11, -7) -- (-5,-5) -- (-5, -7) --  (-5,-5) -- (0, -7);
\draw[thick] (5,-5) -- (11, -7) -- (5,-5) -- (5, -7) --  (5,-5) -- (17, -7);

\end{tikzpicture}
}
}
\caption{Protocol in the form in lemma \ref{bc-protocols-intermediate-form}.}
\label{fig:intermediate-step2}
\end{figure}
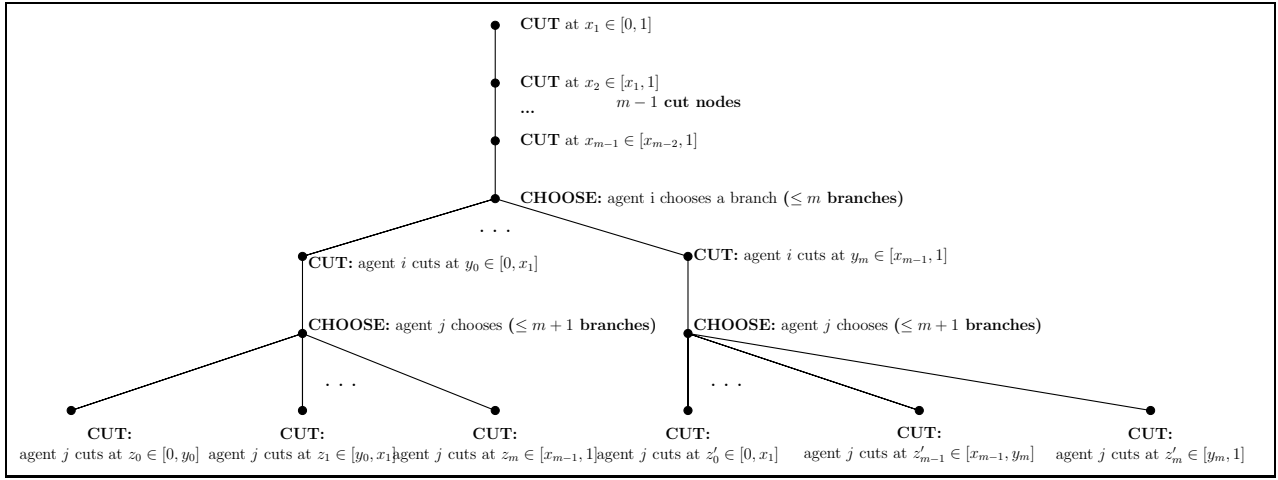

We can assume that we are in the case where we always have the maximum number of branches at each choose node, i.e. all the ``$\leq$'' signs above are equalities. In particular, that means that at a choose node with cut node children, we will have one cut node corresponding to each contiguous piece of cake, so if there are $n$ cuts $x_1...,x_{n}$ when we reach the choose node, it will have cut node children corresponding to cuts into $[0,x_1], [x_1,x_2]...$ up to $[x_{n}, 1]$. 

We can make this assumption because if we want to restrict the cut to the interval $[x_i, x_j]$, we can just remove the branches corresponding to making a cut outside of that interval at the end of the protocol.

We now aim to convert this protocol to one where no choose node has a cut node descendant, by moving cut node children above their choose node parent one by one, as in figures \ref{fig:before} and \ref{fig:after}. 

The difficulty here is that unlike with extended BC protocols in theorem \ref{cuts-before-choices}, when we move a cut node above its choose node parent, we need to modify the other branches of the choose node to account for the extra cut, so that the resulting protocol is still a BC protocol.

Suppose there are $k$ levels of cut nodes below the topmost choose node.

If $k=0$, there are no cut nodes below the topmost choose node, so we are done.

For $k>0$, we start with the topmost level of cut nodes, which will be the children of the topmost choose node. There are $m$ cut nodes ``agent $i$ cuts at $y_j$ into $[x_j, x_{j+1}]$'', $j \in \{0,...,m\} $, where $x_0=0, x_{m} = 1$.

We start by moving the leftmost cut node, ``agent $i$ cuts at $y_0 \in [0,x_1]$'', above its parent. We do not need to modify the branch that cut node was in, but for the other $m-1$ branches, we need to do the following in each branch:
\begin{itemize}
\item the topmost cut node is unchanged;
\item on the second level of cut nodes, we need to split the cut node ``agent $j$ cuts at $z_0' \in [0, x_1]$'' into two cut nodes corresponding to cuts in $[0,y_0]$ and $[y_0, x_1]$. 
\item on the following level, we need to split exactly one cut node in each of the $m+2$ branches to account for the extra cut at $y_0$, etc.
\end{itemize}

In total, this will add $0$ cut nodes on the first level, $m-1$ on the second level, $(m-1)(m+1)$ on the second level,..., $(m-1)(m+1)(m+2)...(m+k)$ on the $k$-th (last) level. In particular, the process terminates.

We repeat this process for the other $m-1$ cut nodes on the first level. At the end, the tree will look like in figure \ref{fig:first-level-moved}, and there will be $k-1$ levels of cut nodes below the top choose node.

\begin{figure}[H]
\centering
\fbox{%
\resizebox{\textwidth}{!}{%
\begin{tikzpicture}

\draw[fill=black] (0,3) circle (3pt);
\draw[fill=black] (0,1.5) circle (3pt);
\draw[fill=black] (0,0) circle (3pt);
\draw[fill=black] (0,-1.5) circle (3pt); 
\draw[fill=black] (-5,-3) circle (3pt);
\draw[fill=black] (5,-3) circle (3pt);
\draw[fill=black] (-11,-5) circle (3pt);
\draw[fill=black] (-5,-5) circle (3pt);
\draw[fill=black] (0,-5) circle (3pt);
\draw[fill=black] (5,-5) circle (3pt);
\draw[fill=black] (11,-5) circle (3pt);
\draw[fill=black] (17,-5) circle (3pt);

\node[anchor=west] at (0.5,3) {\textbf{CUT}};
\node[anchor=west] at (0.5,1.5) {\textbf{CUT}};
\node[anchor=west] at (0.5,0.75) {\textbf{...}};
\node[anchor=west] at (0.5,0) {\textbf{CUT}};
\node[anchor=west] at (3, 1) {\textbf{$2m-1$ cut nodes}};

\node[anchor=west] at (0.5, -1.5) {\textbf{CHOOSE:} agent i chooses a branch \textbf{($m$ branches)}};
\node[align=center, anchor=south] at (0, -2.5) {\textbf{.    .    .}};

\node[anchor=west] at (-4, -3) {\textbf{CHOOSE:} agent $j$ chooses \textbf{($2m+1$ branches)}};
\node[anchor=west] at (6, -3) {\textbf{CHOOSE:} agent $j$ chooses \textbf{($2m+1$ branches)}};

\node[align=center,anchor=south] at (-10, -6.5) {\textbf{CUT:} 
\\
agent $j$ cuts ...};
\node[align=center, anchor=south] at (-5, -6.5) {\textbf{CUT:} 
\\
agent $j$ cuts ...};
\node[align=center, anchor=south] at (0, -6.5) {\textbf{CUT:}
\\
agent $j$ cuts ...};

\node[align=center, anchor=south] at (-4, -4.5) {\textbf{.    .    .}};

\node[align=center, anchor=south] at (5, -6.5) {\textbf{CUT:} 
\\
agent $j$ cuts ...};

\node[align=center, anchor=south] at (6, -4.5) {\textbf{.    .    .}};

\node[align=center, anchor=south] at (11, -6.5) {\textbf{CUT:} 
\\
agent $j$ cuts ...};

\node[align=center, anchor=south] at (17, -6.5) {\textbf{CUT:} 
\\
agent $j$ cuts ...};

\draw[thick] (0,3) -- (0,1.5) -- (0,0) -- (0,-1.5) -- (-5,-3) -- (0, -1.5) -- (5,-3);
\draw[thick] (-5,-3) -- (-5, -5);
\draw[thick] (5,-3) -- (5, -5);
\draw[thick] (-5,-3) -- (-11, -5) -- (-5,-3) -- (-5, -5) --  (-5,-3) -- (0, -5);
\draw[thick] (5,-3) -- (11, -5) --  (5,-3) -- (17, -5);

\end{tikzpicture}
}
}
\caption{Protocol after moving first level of cut nodes above choose node.}
\label{fig:first-level-moved}
\end{figure}
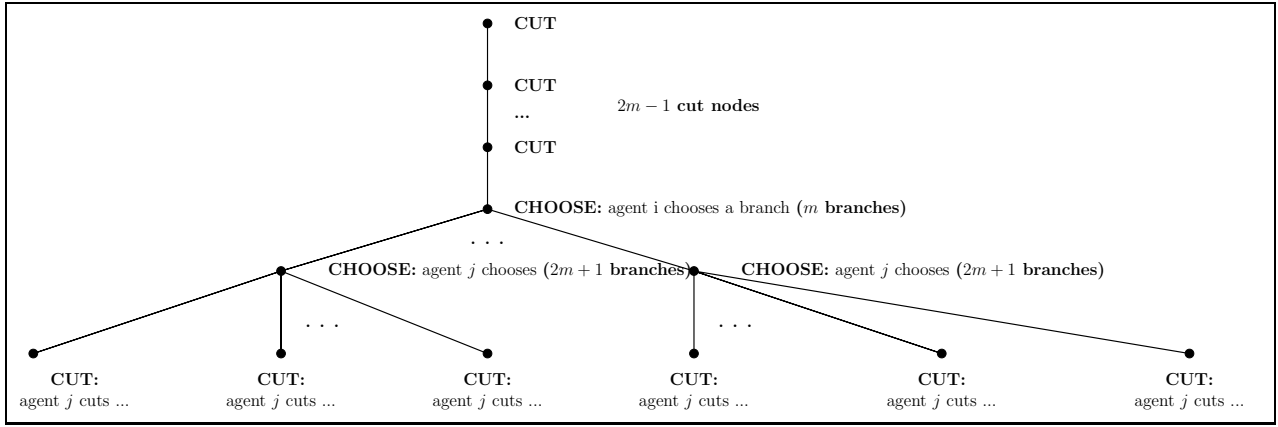

Now, each of the $m$ choose nodes below the top choose node have $2m+1$ cut node children. Starting from the leftmost of these choose nodes and its left most cut node child, we move the child above the parent, and then move it above the top choose node. We repeat this process until we have moved all $m(2m+1)$ cut nodes on this level above the top choose node.

Now there are $k-2$ levels of cut nodes below the top choose node. Recursively, we can move the cut nodes on each level above the top choose node. Since the number of levels decreases at each step, and the process for a single level terminates, the whole conversion process terminates. 

At the end, we also need to ``fix'' the allocation nodes to account for the extra cuts, by changing an allocation $[x,y] \rightarrow j$ to $[x, x_0] \rightarrow j,..., [x_N, y] \rightarrow j$, where $x_0,...x_N$ are all the cuts between $x$ and $y$. This takes time $O(ac)$ where $a$ is the number of allocation nodes and $c$ the number of cut nodes in the final tree.

The resulting protocol will have no cut nodes before the top choose node, so clearly it will contain no choose node with a cut node descendant. 

Note that pulling a cut node above its choose node parent and splitting cut nodes in the other branches as above does not affect the allocation, hence the protocols before and after each step are strongly envy-equivalent. Therefore the original protocol and the final protocol are strongly envy-equivalent.
\end{proof}

\subsection{Equivalence to GCC Protocols}\label{sec:bcgcc}

In this section, we will prove that GCC and BC protocols can express the same class of protocols, up to strong envy-equivalence.

\begin{theorem}\label{gcc-to-bc}
For any GCC protocol, there is a strongly envy-equivalent BC protocol.
\end{theorem}

\begin {proof}
We will refer to a tree representation of a GCC protocol as a  \emph{GCC-tree} (with GCC-cut, GCC-choose, and GCC-if-else nodes), and to a tree representation of a BC protocol as \emph {BC-tree} (with BC-cut and BC-choose nodes, and BC-leaves/allocation nodes).

Consider a protocol that can be represented as a GCC-tree T. We will construct a BC-tree representation of it by replacing each node $v$ in T and its subtree with an equivalent BC-subtree, in topological order. Before we start this process, we connect a BC allocation node to each leaf of T (with the actual allocation to be filled in later). 
\begin{itemize}
\item \textbf{If $v$ is a GCC-cut node:}  At node $v$, agent $i$ chooses a piece from a set $S$ and cuts into it. Since $S$ contains disjoint pieces, we can always extend it to a partition of the cake $S'$. We replace node $v$ with a BC-choose node, with $|S|$ children, so there is one child for each piece $P \in S$. The child corresponding to a piece $P$ will be a BC-cut node representing the action ``agent $i$ cuts into piece $P$''. The subtree of each child in the BC-tree will be a copy of $v$'s subtree. In the BC-choose node, the agent can choose which child to proceed to, so they effectively get a choice of which piece from $S$ to cut into.

\item \textbf{If $v$ is a GCC-choose node:} At node $v$, agent $i$ chooses a piece from a set $S$ and that piece gets assigned to them. Once again, we replace node $v$ with a BC-choose node with $|S|$ children, one for each piece $P \in S$, and make the subtree of each child a copy of $v$'s subtree. The agent will choose to progress to the child corresponding to the piece they would have chosen in the GCC protocol. We now need to ensure that piece $P$ actually gets assigned to agent $i$ in the BC protocol. We add ``agent $i$ gets piece $P$'' to each BC-leaf in the subtree corresponding to piece $P$. At the end of the conversion process, each BC-leaf will contain a valid allocation of the entire cake.

\item \textbf{If $v$ is a GCC-if-else node:} By definition, the conditions in the if-else statement at $v$ depend on  which pieces the agents cut/chose at each of the previous steps of the protocol. But since we have already converted all of $v$'s ancestors to BC nodes, the decisions the agents made at each step before the current node are already determined, so we know what the result of the if-else statement will be and which branch the protocol will progress to. Hence, we remove all the other branches, then delete the if-else node $v$ and connect the single remaining branch to $v$'s parent.
\end{itemize}

Note that after we have converted every node $v$ in the GCC-tree, the tree will only branch out at BC-choose nodes, and progression to a child will only depend on the choice the agent makes at that node. 

The protocol represented by the resulting BC-tree essentially simulates every action in the initial protocol using BC-tree nodes. In particular, every agent will get the same piece of cake they would have gotten in the initial protocol. Therefore, any simultaneous bounds on their envy levels towards other agents in the initial GCC protocol carry over to the BC protocol. Hence the two protocols are strongly envy-equivalent. 
\end{proof}

\begin{remark}\label{gcc-extensive-definition}
The definition of choose/cut nodes in GCC protocols is ambiguous as to whether the set $S$ of pieces of cake an agent is given is restricted to continuous pieces (i.e. pieces $[x, y]$ such that there is no other cut between $x$ and $y$), or we also allow for pieces that contain previous cuts.

Note that the above proof assumes the former, but the result also holds for the more extensive definition, since, by a similar argument to the one in the proof of theorem \ref{extended-bc-protocols}, if a piece $[x, y]$ contains cuts $x_1, ..., x_n$ we can split it into $n+1$ pieces $[x, x_1], ..., [x_n, y]$ and add an if-else node below to establish which of these $n+1$ pieces the agent acted on. This way, we can convert a protocol in the more extensive definition to the restricted one, and then convert it to a BC protocol via theorem \ref{gcc-to-bc}.
\end{remark}

\begin{theorem}\label{bc-to-gcc}
For any BC protocol, there is a strongly envy-equivalent GCC protocol.
\end{theorem}

\begin{proof}

We start by putting the BC protocol in the ``cuts before choices'' form given by theorem \ref{bc-cuts-before-choices}. 

The BC-cut nodes can then be trivially converted to GCC-cut nodes. The BC-leaves can be converted to a chain of GCC-choose nodes (so if agent $i$ is allocated some interval $[x, y]$ in the original node we make them ``choose'' the piece $[x, y]$ instead. If agent $i$ gets multiple intervals, we will have one GCC-choose node for each).

It remains to simulate the BC-choose nodes (in which an agent chooses which branch of the tree to proceed to) using GCC nodes. We will do this by creating a piece that has value $0$ for every agent, and then dividing that piece between all $n$ agents. Each agent will then further divide their piece and ``choose'' between various sub-pieces of it to artificially simulate branch choices. 

The actual protocol will be run on the remaining cake, which will be an interval $[a_n, 1], a_n \geq 0$ such that $V_i([a_n,1]) = 1$ for every $i$. We can do this by replacing the endpoint $0$ with $a_n$ in every node of the protocol. Clearly any envy bounds in the original protocol will hold in the new protocol, since the piece $[0, a_n]$ that we cut off has no value for any of the agents.

Explicitly, we add the following chain of cuts at the beginning of the protocol:

\begin{itemize}
\item agent $1$ cuts into $[0,1]$ at point $a_1$
\item agent $2$ cuts into $[0,a_1]$ at point $a_2$

...
\item agent $n$ cuts into $[0,a_{n-1}]$ at point $a_n$
\end {itemize}

The optimal strategy for agent $i$ is to cut at some $a_i$ such that $V_i([0, a_i]) = 0$, so the piece $[0, a_n]$ will have value $0$ for all agents, as required. This is because the agents cannot control how much of the piece $[0, a_n]$ they get, so they want it to have as little value for them as possible.

\begin{itemize}
\item agent $1$ cuts into $[0,a_n]$ at point $b_1$
\item agent $2$ cuts into $[0,b_1]$ at point $b_2$

...
\item agent $n-1$ cuts into $[0,b_{n-2}]$ at point $b_{n-1}$
\end{itemize}

Now, agent $i$ will use the piece $[b_i, b_{i-1}]$ (where $b_n = 0, b_0 = a_n$) to simulate their BC-choose nodes as follows:

If there are $m$ BC-choose nodes controlled by agent $i$ in the tree, they cut their piece $m-1$ times, dividing it into $m$ pieces. Each piece will correspond to a choose node.

If the BC-choose node corresponding to the $j$-th piece has $k$ branches, we replace it with the following process:
\begin{itemize}
\item Agent $i$ divides the $j$-th piece into $k$ sub-pieces, one for each branch.
\item In a GCC-choose node, agent $i$ chooses one of the $k$ sub-pieces.
\item In a GCC-if-else node, we proceed to one of the $k$ branches, depending on which piece the agent chose.
\item The agent gets the other $k-1$ pieces one by one (using a GCC choose node for each one).
\end{itemize} 

Note that not all of agent $i$'s $m$ BC-choose nodes might be reached in a specific run of the protocol, since they might be in different branches. This means that in some branches, the pieces corresponding to those choose nodes will not have been allocated to anyone. To resolve this issue, agent $i$ gets the remaining sub-pieces of $[b_i, b_{i-1}]$ at the end.

This way, agent $i$ always gets the entire interval $[b_i, b_{i-1}]$, so the interval $[0, a_n]$ is all allocated to various agents. 

Note there are some extra cuts in the cake at $a_1, ..., a_{n-1}$ that we have ignored in the protocol. This is fine in the more extensive interpretation of GCC protocols, and it can be converted to the more restrictive interpretation by the Remark above.

The resulting protocol will be strongly envy-equivalent to the initial protocol, since, as explained above, none of the steps in the conversion process affect the envy bounds.
\end{proof}

\begin{corollary}\label{extended-bc-to-gcc-polynomial}
If we allow for the more extensive interpretation of GCC protocols (as in remark \ref{gcc-extensive-definition}), any extended BC protocol can be converted to a strongly-envy equivalent GCC protocol in polynomial time.
\end{corollary}
\begin{proof}
Consider the conversion algorithm given in the proof of theorem \ref{bc-to-gcc}. Apart from converting the BC protocol to an equivalent BC protocol in ``cuts before choices'' form at the start, and converting the resulting GCC protocol to fit the more restrictive interpretation of the definition at the end, the rest of the algorithm runs in polynomial time, since we only add a polynomial number of nodes to the tree.

Therefore, if we use the more extensive interpretation of GCC protocols instead, we can use theorem \ref{cuts-before-choices} to convert the original protocol to an extended BC protocol in ``cuts before choices'' form in polynomial time, then run the conversion algorithm as before.
\end{proof}

Note this argument does not work for the reverse conversion (GCC to BC) given by theorem \ref{gcc-to-bc}, because there we specifically make use of the more restrictive classes of protocols when converting, for example, GCC if-else nodes.

\subsection{Proportional Cake-Cutting Protocols}\label{sec:bcprop}
A cake-cutting protocol for $n$ agents is said to be proportional if each agent is guaranteed to receive a piece of cake with value at least $\frac{1}{n}$ according to their own valuation function. There are a number of well-known proportional cake-cutting protocols, such as the Dubins-Spanier protocol \citep{dubins-spanier} and the Even-Paz protocol \citep{even-paz}. 

In this section, we look at how these can be represented as GCC and BC protocols.

The discrete Dubins-Spanier protocol works as follows: each agent $i$ makes a cut at a point $a_i$ such that $V_i([0, a_i]) = \frac{1}{n}$. The agent who made the leftmost cut $a_j$ gets the piece $[0, a_j]$. That agent is then removed and we repeat the protocol on the remaining cake $[a_j, 1]$ with the other $n-1$ agents. The last agent gets the remaining piece of cake.

\citet{branzei-thesis} describes how the Dubins-Spanier protocol can be implemented as a GCC protocol. Each agent is asked to make a cut in $[0,1]$ at some point $a_i$. The leftmost cut $a_j$ is then determined using an If-Else node with $n$ branches, then, in the $j$-th branch, agent $j$ ``chooses'' a piece from the singleton set $\{[0, a_j]\}$, and we repeat the procedure for the other $n-1$ agents on the remaining cake, ignoring the cuts made at previous steps. This protocol clearly respects the proportionality condition, so, in our language, it is value-equivalent to the original protocol. Any bounds on envy from the original protocol would also clearly carry over, so they are strongly envy-equivalent. Note that the size of the resulting GCC tree is $\Theta(n!)$, because there are $n$ branches on the first if-else node, then $n-1$ branches for each of those, etc. The runtime of the protocol is $\Theta(n^2)$ if we just count the queries/nodes, as in the original protocol, and $\Theta(n^3)$ if we also consider the comparisons inside each if-else node, since to find the leftmost cut we need to compare the position of each cut to the other $n-1$. 

We can implement this as a BC protocol either by converting the GCC implementation using theorem \ref{gcc-to-bc}, or as an extended BC protocol as follows: agent $1$ makes a cut in $[0,1]$ at $a_1$, agent $2$ chooses between two branches (one for cutting in $[0,a_1]$ and the other one for $[a_1, 1]$), the next agent chooses between two branches corresponding to cutting left or right of the current leftmost node. After this, we know who made the leftmost cut $a_j$ in a specific branch. At the end of the protocol agent $j$ will get allocated the piece $[0, a_j]$. We repeat the algorithm for the other $n-1$ agents on the remaining cake, ignoring the cuts made at previous steps. There are $2^{n-1}$ branches in the first step, then for each of those there are $2^{n-2}$ branches in the second step, etc. So the size of the tree is $2^{n-1}2^{n-2}...2^1 = 2^{\frac{n(n-1)}{2}}$, so the size of the protocol is $\Theta(2^{n^2})$.

\bigskip
The Even-Paz protocol for $n$ agents, is as follows:
\begin{itemize}
\item Each agent makes a cut at $z_i \in [0,1]$ such that $V_i([0, z_i]) = \frac{1}{2}$.
\item The algorithm finds the $\lfloor \frac{n}{2} \rfloor$-th cut, which divides the cake into two ``halves''.
\item Each half is divided recursively amongst the agents whose cuts were inside that half (with the agent who made the $\lfloor \frac{n}{2} \rfloor$-th cut going into the first half).
\end{itemize}

The runtime is $\Theta(n \log(n))$.

We can implement the Even-Paz protocol as a GCC protocol as follows: each agent makes a cut at $z_i \in [0,1]$, then in an if-else node with $n$ branches, we identify the $\lfloor \frac{n}{2} \rfloor$-th cut. We repeat the procedure for each half for the $n/2$ agents who cut into that half. There are $n$ branches in the first step, then in each of those we have $n/2$ branches for dividing the first half, each of which then splits into $n/2$ branches for dividing the second half, etc.
Assuming $n=2^k$ for simplicity, the size of the tree is $n (\frac{n}{2})^2 (\frac{n}{4})^2 ... = 2^k 2^{2(k-1)} 2^ {2(k-2)} ... 2^2 = 2^{k + 2 \frac{k(k-1)}{2}} = 2^{k^2}$, so $\Theta(2^{k^2}) = \Theta(n^{\log(n)})$, and the resulting protocol is value-equivalent (so, in particular, proportional) and strongly envy-equivalent to the original one.

For the extended BC protocol implementation, agent $1$ makes a cut at $a_1$ in $[0,1]$, agent $2$ chooses between two branches (for cutting in $[0, a_1]$ or $[a_1, 1]$ respectively), agent $3$ chooses between three branches (corresponding to cutting into each of the three existing pieces), etc. After this, we know the ordering of the cuts in a specific branch, so, in particular, the $\lfloor \frac{n}{2} \rfloor$-th cut. We recurse on the two halves of the cake.
Assuming $n = 2^k$, we get $n!$ branches in the first step, $((\frac{n}{2})!)^2 = \Theta(n!)$ branches in the second step, etc, so the size of the tree is $n! ((\frac{n}{2})!)^2 ((\frac{n}{4})!)^2 ... $

\section{Conclusions and Future Work}
\subsection{Conclusion}
In this paper, we defined a new model for representing cake-cutting protocols as trees, called branch choice (BC) protocols, which differs from other models in that instead of choosing specific pieces of cake, the agents get to choose which branch of the tree to proceed to at a certain point, and then the cake is allocated at the end. We showed that various modifications to the model do not impact its expressive power. We also proved that any BC protocol can be converted to an equivalent protocol in which the cut nodes come before the choose nodes, so that, informally, the agents can first cut up the cake, and then choose between various branches that lead to different allocations.

The main benefit of BC protocols is that they form a rather bare-bones model, so that, for example, if the protocol reaches a given node, we can deduce the execution history of the protocol (the relative order of the cuts, as well as which branches the agents chose) just based on the path from the root to that node, with no ambiguity about what decisions the agents might have made at previous nodes. This makes it easier to reason about the protocols, particularly after putting them into the special ``cuts before choices'' form mentioned above, since in that case we also know a lot about the structure of the tree. 

This comes at no cost to the expressive power of the model, since the class of protocols it covers is the same as generalised cut and choose protocols as described by \citet{branzei}, up to strong envy equivalence (which we defined to mean that in both protocols, an agent can guarantee the same simultaneous bounds on their envy against any subset of the other agents).

However, as we have seen both in some of our conversion algorithms and in the case of representing some classic protocols as BC protocols, the simplicity of the model comes with a trade-off in space complexity, and in some cases the conversion seems to result in $\Theta(n!)$ or even worse blowups in the size of the tree representation compared to the original protocol.

\subsection{Future Work}

\begin{paragraph}{Checking equivalence} 
Most of our results involve constructing protocols that are equivalent to a given protocol. It would be interesting to look into the computational complexity of checking whether two given BC protocols are equivalent under one of the notions of equivalence we have defined, or the highly related question of computing the value/envy bounds an agent can guarantee for themselves in a given protocol. Note that the former question very easily reduces to the latter. 
\end{paragraph}

\begin{paragraph}{Simplifying protocols}
As we have seen above, explicitly representing a cake-cutting protocol as a tree, whether in the GCC or the BC protocols model, often results in a blowup in the size of the protocol, compared to its original statement. A natural question to ask is whether we can simplify a given protocol, i.e. find a strongly envy-equivalent protocol with a smaller tree representation. A brute-force approach to this would involve finding all protocols whose tree representation is smaller than the given one, and then checking, for each of these protocols, whether it is equivalent to the original protocol (which reduces to the question discussed previously).
\end{paragraph}

\begin{paragraph}{Adding extra structure to the model}
Another possible direction is to consider whether we can decrease the size of the tree representation by adding some extra structure to the model, as we already did informally in Algorithm \ref{selfridge-conway-gcc} for the Selfridge-Conway protocol, where on line \ref{label-piece-a}, for example, we label the piece agent $2$ cut $A$ and the other two pieces $B$ and $C$, instead of explicitly showing the three (virtually identical) branches corresponding to each piece. 

We propose two possible types of nodes to add:
\begin{itemize}
\item \emph{Piece permutation:} agent $i$ permutes the order of the existing pieces of cake. If there are currently $m$ cuts in the cake, then there are $m+1$ continuous pieces which we can number $1,..., m$ from left to right, and agent $i$ can apply some permutation $\sigma \in S_{m+1}$ to reorder them. Note this is functionally equivalent to the labelling used in the Selfridge-Conway protocol, which we mentioned above, since agent $2$ could, for example, reorder the pieces in increasing order of their value, and then cut the third piece. This would reduce the size of the BC tree representation of the Selfridge-Conway protocol in figure \ref{fig:selfridge-conway-bc}, since we would only need a single branch at the topmost choose node instead of three. This leads us to believe there is a class of protocols for which adding this type of node can be proven to reduce the size of the protocol, possibly exponentially in the case where this kind of node is used multiple times.

Adding piece permutation nodes to the BC protocols model is straightforward, as long as we refer to the pieces by their position in the cake (e.g. ``agent $1$ cuts the second piece from the left''), rather than by their endpoints, which in the case of BC protocols does not make any difference.

We can also add them to GCC protocols, provided we change the model to refer to pieces by their position in the cake, but in this case we have to decide whether we should account for pieces that agents have already taken from the cake, e.g. if the cake is divided into three pieces $1,2,3$ and some agent chooses piece $2$, should the remaining pieces now be $1$ and $3$ or can we relabel them $1$ and $2$?

\item \emph{Agent permutation:} There are two types of action we can consider here: either some agent $i$ permutes the order of the other agents, which would allow them to choose, for example, which of the other agents they are envious of in some situation, or we have a referee who can reorder the agents according to some rule. The former option can be added to both GCC and BC protocols, while the latter can more naturally be added to GCC protocols, as an extension to the existing if-else node, and could be used, for example, in the case of the Dubins-Spanier protocol, where we use an if-else node to identify the agent who made the leftmost cut, in which case we can reorder the agents so that agent is now the $n$-th one, and recurse on the first $n-1$ agents.
\end{itemize}
\end{paragraph}

\begin{paragraph}{Converting other protocols to the BC model}
We have already shown that some classic protocols (Selfridge-Conway, Dubins-Spanier, Even-Paz) can be converted to BC protocols, but it would be interesting to extend this result to other protocols in the literature. Of particular interest would be the discrete bounded envy-free protocol for $n$ agents proposed by \citet{aziz}. We suspect that it is possible to convert this protocol to a BC protocol, since a lot of the additional structure used in the protocol could be encoded into the BC tree representation, or implicitly in the agents' strategies to achieve envy-freeness, but we have not checked this.
\end{paragraph}

\section*{Acknowledgements}
We thank Simina Br\^{a}nzei for helpful comments.

\bibliographystyle{authordate1}
\bibliography{citations}

\end{document}